%% file: main.tex
\documentclass[11pt]{article}
\input{pretex}

\usepackage{amsmath,amsfonts}
\usepackage{algorithmic}
\usepackage{algorithm}
\usepackage{array}
\usepackage[caption=false,font=normalsize,labelfont=sf,textfont=sf]{subfig}
\usepackage{textcomp}
\usepackage{stfloats}
\usepackage{url}
\usepackage{verbatim}
\usepackage{graphicx}
\usepackage{cite}
\usepackage{mathrsfs}
\usepackage{soul}
\usepackage{subfig}
\usepackage{authblk} 

\definecolor{colortwo}{rgb}{0.4,0.77,0.17}
\definecolor{colorthree}{rgb}{0.01,0.51,0.93}
 
\newcommand*\samethanks[1][\value{footnote}]{\footnotemark[#1]}

\allowdisplaybreaks

\begin{document}
\title{Entanglement cost of discriminating quantum states under locality constraints}

\author[1]{Chenghong Zhu\thanks{Chenghong Zhu and Chengkai Zhu contributed equally to this work.}}
\affil[1]{\small Thrust of Artificial Intelligence, Information Hub,\par The Hong Kong University of Science and Technology (Guangzhou), Guangdong 511453, China}

\author[1]{Chengkai Zhu\samethanks[1]}

\author[1,2]{Zhiping Liu}

\author[1]{Xin Wang\thanks{felixxinwang@hkust-gz.edu.cn}}

\affil[2]{National Laboratory of Solid State Microstructures, School of Physics and Collaborative Innovation Center of Advanced Microstructures, Nanjing University, Nanjing 210093, China}

\date{\today}
\maketitle

\begin{abstract}
The unique features of entanglement and non-locality in quantum systems, where there are pairs of bipartite states perfectly distinguishable by general entangled measurements yet indistinguishable by local operations and classical communication, hold significant importance in quantum entanglement theory, distributed quantum information processing, and quantum data hiding. This paper delves into the entanglement cost for discriminating two bipartite quantum states, employing positive operator-valued measures (POVMs) with positive partial transpose (PPT) to achieve optimal success probability through general entangled measurements. We first introduce two quantities called the spectral PPT-distance and relative spectral PPT-distance of a POVM to quantify the localness of a general measurement. We show these quantities are related to the entanglement cost of optimal discrimination by PPT POVMs.
Following this, we establish bounds and develop SDP hierarchies to estimate the entanglement cost of optimal discrimination by PPT POVMs for any pair of states. Leveraging these results, we show that a pure state can be optimally discriminated against any other state with the assistance of a single Bell state. This study advances our understanding of the pivotal role played by entanglement in quantum state discrimination, serving as a crucial element in unlocking quantum data hiding against locally constrained measurements.
\end{abstract}

\tableofcontents

\section{Introduction}
In the realm of quantum information, nonlocality manifests when local measurements on a multipartite quantum system may not be able to reveal in which state the system is prepared, even among mutually orthogonal product states~\cite{Bennett1999b}. It has consistently been a crucial research direction to understand the power and limits of quantum operations that can be implemented by local operations and classical communication (LOCC), e.g., in quantum state discrimination~\cite{Leung_2021,Bandyopadhyay2015,Childs2013,Bandyopadhyay2011a,Calsamiglia2010,Halder2019,Walgate2000,Bennett1999b, chitambar2013local, chitambar2013revisiting}.

In quantum state discrimination tasks, one usually performs a binary-valued positive-operator valued measure (POVM) on the received state and then decides which state it is according to the measurement outcome. Quantum state discrimination has led to fruitful applications in quantum cryptography~\cite{Gisin_2002, Cleve_1999, Leverrier_2009}, quantum dimension witness~\cite{Brunner_2013, Hendrych_2012} and quantum data hiding~\cite{Terhal2001,Eggeling2002,Matthews_2009}. 
To understand the underlying mechanism of quantum state discrimination, various operations that encompass LOCC are further considered due to its complex structure, including POVMs that are separable (SEP)~\cite{Bandyopadhyay2015LimSep, chitambar2009nonlocal}, and POVMs with positive partial transpose (PPT)~\cite{Yu2014,li2017indistinguishability, cheng2023discrimination}. There is a strict inclusion among them as~\cite{Bennett1999b},
\begin{equation}\label{Eq:povm_hire}
 \text{LOCC} \subseteq \textrm{SEP} \subseteq \textrm{PPT} \subseteq \textrm{ALL}.
\end{equation}
Notably, PPT POVMs enjoy a much simpler mathematical structure than LOCC and SEP as the former can be completely characterized by semidefinite programming (SDP)~\cite{boyd2004convex}. 
In particular, the distinguishability of quantum states under a restricted family of measurements is closely related to quantum data hiding~\cite{Takagi_2019,Lami_2018}, which is originally based on the existence of pairs of bipartite states that are perfectly distinguishable with general entangled measurements yet indistinguishable by LOCC. In the context of bipartite state discrimination, the utility of a shared entangled state is paramount as it can significantly improve limited local distinguishability, potentially matching the efficacy of general entangled measurements. Notably, a maximally entangled state is often a universal resource for such enhancements~\cite{Cohen2008, Bandyopadhyay2015}. This leads to a compelling and crucial question:
\begin{center}
    \emph{How much entanglement at minimum do we need to unlock this quantum data hiding?}
\end{center}

Due to the complex structure of LOCC and entanglement, our understanding of the role of entanglement in quantum state discrimination is still limited.
Currently, the entanglement cost is only known for a small number of ensembles of quantum states, with a focus on the set that contains at least one of the four Bell states. Specifically, it has been shown that Bell basis~\cite{Ghosh2001Bellbasis} and a set of three Bell states~\cite{Bandyopadhyay2015LimSep} can be distinguished with the assistance of 1 ebit.
For quantum states with higher dimensions, Ref.~\cite{Yu2014} shows that 1 ebit suffices for discriminating a $d\ox d$ pure state and its orthogonal complement perfectly by PPT POVMs.
For other special cases, Ref.~\cite{Bandyopadhyay_2021} demonstrates that 1 ebit is required for the optimal discrimination for a noisy Bell state ensemble by LOCC.

In this paper, we study the entanglement cost of optimally discriminating two arbitrary quantum states by PPT POVMs, which is also the lower bound for the scenario of employing LOCC. {Firstly, we introduce the spectral PPT-distance and the relative spectral PPT-distance of a POVM which characterize the PPT-ness of a general POVM.} Secondly, we formally define the entanglement cost of optimal discriminating a set of states via measurements under locality constraints. {We show that the spectral PPT-distance and the relative spectral PPT-distance of a POVM can be utilized to characterize the entanglement cost of optimal discrimination by PPT POVMs. Thirdly, we provide an upper bound on the entanglement cost of optimal discrimination by PPT POVMs which is directly determined by the spectrum properties of the given states. Leveraging this upper bound, we introduce an SDP hierarchy to estimate the entanglement cost of optimal PPT-discrimination through iterative steps. As examples,} we show that a pure state can always be optimally discriminated against any other states with the assistance of 1 ebit, and provide a tighter estimation of the entanglement required compared with teleportation when discriminating orthogonal mixed states.

\begin{figure}[t]
    \centering
    \includegraphics[width=0.6\linewidth]{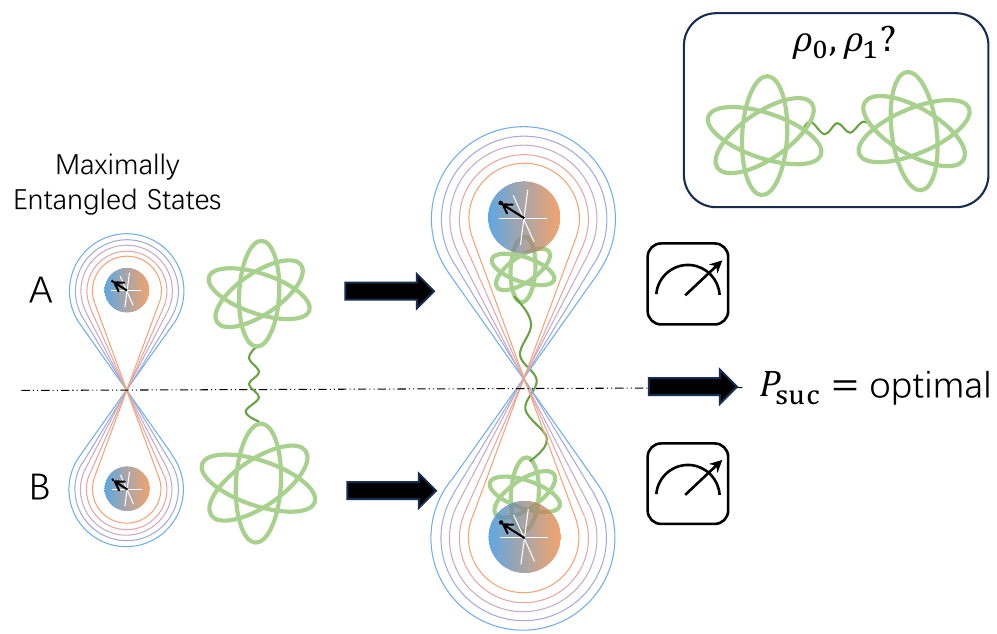}
    \caption{Illustration of the entanglement-assisted quantum state discrimination task. Alice and Bob (represented as $A$ and $B$) aim to discriminate bipartite quantum states $\rho_0$ and $\rho_1$. Commencing with an initial allocation of entanglement resources, the optimal discrimination achievable through a general entangled measurement is demonstrated to be attainable by employing measurements under locality constraints, showcasing the pivotal role of entanglement consumption in the discrimination process.}
    \label{fig:enter-label}
\end{figure}

\section{Preliminaries}

\paragraph{Notations}
We denote the finite-dimensional Hilbert spaces associated with the systems of Alice and Bob as $\cH_A$ and $\cH_B$, respectively. The dimensions of $\cH_A$ and $\cH_B$ are denoted as $d_A$ and $d_B$. A quantum state on system $A$ is a positive semidefinite operator $\rho_A$ with trace one. {We denote by $\cL(\cH_A)$ the set of all linear operators on $\cH_A$ and denote by $\cD(\cH_A)$ the set of all quantum states (density matrices) on $\cH_A$.} Let $\{\ket{j}\}_{j = 0,\cdots,d-1}$ be a standard computational basis, then a standard maximally entangled state of Schmidt rank $d$ is $\Phi_d^+ = 1/d \sum_{i,j = 0}^{d-1}\ketbra{ii}{jj}$. A bipartite quantum state $\rho_{AB}$ is said to be Positive-Partial-Transpose (PPT) if $\rho_{AB}^{T_B} \geq 0$ where $T_B$ denotes taking partial transpose on the subsystem $B$.

\paragraph{Quantum State Discrimination}
Quantum measurements are described by positive-operator-valued-measures (POVMs). A POVM with $n$ elements is denoted as $\boldsymbol{E} = \{E_j\}_{j=0}^{n-1}$ where $E_j \geq 0$ and $\sum_j E_j = I$. Suppose $\Omega = \{(p_j, \rho_j)\}_{j=0}^{n-1}$ is an ensemble of $n$ quantum states $\rho_0 \dots \rho_{n-1}$ given with associated probability $p_0 \dots p_{n-1}$. In the task of quantum state discrimination, we usually apply a POVM on a given unknown state $\rho_j$ with prior probability $p_j$ and decide which state it is according to the measurement outcomes. In the context of minimum-error state discrimination (MSD), the goal is to maximize the average success probability, given by
\begin{equation}\label{Eq:suc_prob}
    P_{\suc}(\Omega) = \max_{\{E_j\}}\sum_{j=0}^{n-1} p_j \tr(\rho_j E_j).
\end{equation}
Notably, the celebrated result in Ref.~\cite{helstrom1969quantum} shows that optimal discrimination success probability of discriminating two quantum states $\{\rho_0, \rho_1\}$ with prior probability $p_0$ and $p_1$ is given by $\frac{1}{2}(1+ \|p_0\rho_0 - p_1\rho_1 \|_1)$, 
where $\| \cdot \|_1$ denotes the trace norm. The corresponding POVM is called the \textit{Helstrom measurements}. It can be expressed by $M_+$ and $M_{-}$, where $M_+$ is the projection on the positive eigenspace of $p_0\rho_0 - p_1\rho_1$ and $M_- = I-M_+$.

Considering POVMs under locality constraints acting on a bipartite system $AB$, we denote {$\mathbf{\LOCC}_{AB}$} as the set of all LOCC POVMs~\cite{Chitambar_2014} which includes strategies based on measurement outcomes performed locally on each subsystem. We denote
\begin{equation}
    \boldsymbol{\rm SEP}_{AB}:=\left\{ \left\{E_k^{(1)}\ox E_k^{(2)}\right\}_k: E_{k}^{(1)}, E_k^{(2)}\in \cL(\cH_A\ox\cH_B),~E_{k}^{(1)}, E_k^{(2)}\geq 0,~\sum_{k}E_k^{(1)}\ox E_k^{(2)} = I\right\}
\end{equation}
as the set of all SEP POVMs~\cite{Matthews2009} and
\begin{equation}
    \boldsymbol{\rm PPT}_{AB}:=\left\{ \{E_{k}\}_{k}: E_{k}\in \cL(\cH_A\ox\cH_B),~E_{k}\geq 0,\, E_{k}^{T_B}\geq 0,\, \sum_k E_{k} = I \right\}
\end{equation}
as the set of all PPT POVMs~\cite{Matthews2009}.

\section{Spectral PPT-distance of a POVM}
To investigate the capability of measurements under locality constraints, we first focus on the characterization of the localness of a POVM. In light of the favorable mathematical characterization of PPT POVMs, we introduce the generalized PPT-distance of a POVM as follows.
\begin{definition}[Generalized PPT-distance]
Given a POVM $\boldsymbol{M} = \{M_j\}_{j=0}^{n-1}$ acting on a bipartite system $\cH_A\ox\cH_B$, its generalized PPT-distance is defined as
\begin{equation}
\mathscr{D}(\boldsymbol{M}) = \min\bigg\{r\,\bigg|\,\Delta(M_j^{T_B}\|E_j^{T_B})\leq r,\, {\{E_j\}\in \boldsymbol{\rm PPT}_{AB}}\bigg\},
\end{equation}
where {$\Delta(\cdot \| \cdot)$ is a quantity that satisfies positive-definiteness, i.e., $\Delta(M_j^{T_B} \|E_j^{T_B}) \geq 0$ and is equal to zero if and only if $M_j^{T_B} = E_j^{T_B}$}. The minimization ranges over all PPT POVMs.
\end{definition}
{The generalized PPT-distance could be understood as a quantity that characterizes the disparity of a POVM to the set $\boldsymbol{\rm PPT}$, as a reminiscent of the relative entropy of PPT entanglement~\cite{AudenaertREPPT,Hayashi2006a,Plenio2007}. In below, we show different choices for $\Delta(\cdot\|\cdot)$ and their associated properties.}

\subsection{Spectral PPT-distance}
An essential choice of $\Delta(\cdot\|\cdot)$ is the distance induced by the Schatten $p$-norms, i.e., 
\begin{equation}
\left\|X\right\|_p = \left(\tr\left[(X^\dag X)^{\frac{p}{2}}\right]\right)^{\frac{1}{p}},\, p\in[1,\infty).
\end{equation}
This family includes the three most commonly used distances in quantum information theory, i.e., the trace distance, the Frobenius distance, and the spectral distance. The Schatten-$p$ PPT-distance is denoted as $\mathscr{D}_{p}(\boldsymbol{M})$ for $\Delta(M_j^{T_B} \| E_j^{T_B}):= \|M_j^{T_B}-E_j^{T_B}\|_p$. Its geometric interpretation is depicted in Fig.~\ref{fig:Ent_cost_dis}. For $p=\infty$ where $\|\cdot\|_{\infty}$ corresponds to a fundamentally important norm of the spectral norm, we have \textit{the spectral PPT-distance of a POVM} as follows.
\begin{definition}[Spectral PPT-distance]
Given a POVM $\boldsymbol{M}=\{M_j\}_{j=0}^{n-1}$ acting on a bipartite system $\cH_A\ox\cH_B$, its spectral PPT-distance is defined as
\begin{equation}
\mathscr{D}_{\infty}(\boldsymbol{M}) = \min\bigg\{\, r\,\bigg|\,\left\|M_j^{T_B}-E_j^{T_B}\right\|_{\infty}\leq r,\, {\{E_j\}\in \boldsymbol{\rm PPT}_{AB}} \bigg\},
\end{equation}
where the minimization ranges over all possible $\{E_j\}_j$.
\end{definition}

The spectral PPT-distance of a POVM bears nice properties. Firstly, it is efficiently computable via SDP and quantifies the proximity of a given POVM to the set of PPT POVMs. Specifically, $\mathscr{D}_{\infty}(\boldsymbol{M})$ is faithful, i.e., $\mathscr{D}_{\infty}(\boldsymbol{M})\geq 0$ and $\mathscr{D}_{\infty}(\boldsymbol{M})=0$ if and only if $\boldsymbol{M}$ is a PPT POVM. Secondly, it is invariant under local unitaries which is directly obtained by the local unitary invariance of the spectral norm and PPT-ness of $E_j$. Thirdly, for two POVMs $\boldsymbol{M}_0 = \{M_{0,j}\}_{j=0}^{n-1}$ and $\boldsymbol{M}_1 = \{M_{1,j}\}_{j=0}^{n-1}$, we denote a convex combination of them as $\widehat{\boldsymbol{M}} = \alpha\boldsymbol{M}_0 + (1-\alpha)\boldsymbol{M}_1:= \{\widehat{M_j}\}_{j=0}^{n-1}$ where $\widehat{M_j} = \alpha M_{0,j} + (1-\alpha) M_{1,j}$; denote the tensor product of them as $\widetilde{\boldsymbol{M}}=\boldsymbol{M}_0\ox\boldsymbol{M}_1:=\{\widetilde{M}_{jk}\}_{j,k=0}^{n-1}$ where $\widetilde{M}_{jk} = M_{0,j}\ox M_{1,k}$. We show that $\mathscr{D}_{\infty}(\boldsymbol{M})$ is convex w.r.t. the convex combination and linearly subadditive w.r.t. the tensor product of POVMs.

\begin{proposition}[Convexity]\label{prop:cov}
For POVMs $\boldsymbol{M}_0, \boldsymbol{M}_1$ acting on a bipartite system $\cH_A\ox\cH_B$ and $0\leq \alpha \leq 1$,
\begin{equation}
    \mathscr{D}_\infty(\alpha\boldsymbol{M}_0 + (1-\alpha)\boldsymbol{M}_1) \leq \alpha \cdot\mathscr{D}_\infty(\boldsymbol{M}_0) + (1-\alpha)\cdot \mathscr{D}_\infty(\boldsymbol{M}_1).
\end{equation}
\end{proposition}
\begin{proof}
Denote the optimal POVMs for $\mathscr{D}_{\infty}(\boldsymbol{M}_0)$ and $\mathscr{D}_{\infty}(\boldsymbol{M}_1)$ are $\{S_{j}\}_j$ and $\{T_j\}_j$, respectively. $\{\alpha S_j + (1-\alpha)T_j\}_j$ is a valid POVM as well. Then we have
\begin{equation}
    \mathscr{D}_{\infty}(\boldsymbol{M}) \leq \max_{j} \;\|\alpha M_{0,j}^{T_B}+ (1-\alpha) M_{1,j}^{T_B}- \alpha S_j^{T_B} - (1-\alpha) T_j^{T_B}\|_{\infty}.
\end{equation}
Notice that 
\begin{equation}
\begin{aligned}
    \|\alpha M_{0,j}^{T_B}+ (1-\alpha) M_{1,j}^{T_B}- \alpha S_j^{T_B} - (1-\alpha) T_j^{T_B}\|_{\infty}
    \leq & \|\alpha (M_{0,j}^{T_B} - S_j^{T_B}) + (1-\alpha) (M_{1,j}^{T_B} - T_j^{T_B})\|_{\infty}\\
    \leq & \alpha \| M_{0,j}^{T_B} - S_j^{T_B}\|_{\infty} + (1-\alpha)\|M_{1,j}^{T_B} - T_j^{T_B}\|_{\infty}\\
    \leq & \alpha\mathscr{D}_\infty(\boldsymbol{M}_0) + (1-\alpha) \mathscr{D}_\infty(\boldsymbol{M}_1).
\end{aligned}
\end{equation}
Hence, we complete the proof.
\end{proof}

\begin{proposition}[{Linear subadditivity}]\label{prop:lin_subadd}
For POVMs $\boldsymbol{M}_0=\{M_{0,j}\}_j,\; \boldsymbol{M}_1=\{M_{1,k}\}_k$ acting on a bipartite system $\cH_A\ox\cH_B$,
\begin{equation}
\mathscr{D}_{\infty}(\boldsymbol{M}_0\ox\boldsymbol{M}_1) \leq \xi \cdot \mathscr{D}_{\infty}(\boldsymbol{M}_0) + \zeta \cdot \mathscr{D}_{\infty}(\boldsymbol{M}_1),
\end{equation}
where $\xi = \max_k (1+\|M_{1,k}^{T_B}\|_{\infty})/2$ and $\zeta = \max_j (1+\|M_{0,j}^{T_B}\|_{\infty})/2$.
\end{proposition}
\begin{proof}
Denote the optimal POVMs for $\mathscr{D}_{\infty}(\boldsymbol{M}_0)$ and $\mathscr{D}_{\infty}(\boldsymbol{M}_1)$ are $\{S_{j}\}_j$ and $\{T_j\}_j$, respectively. Then we have
\begin{equation}
\mathscr{D}_{\infty}(\boldsymbol{M}) \leq \max_{j,k} \;\|M_{0,j}^{T_B}\ox M_{1,k}^{T_{B'}} - S_j^{T_{B}}\ox T_k^{T_{B'}}\|_{\infty}.
\end{equation}
Notice that 
\begin{equation}\label{Eq:add_ineq0}
\begin{aligned}
    \|M_{0,j}^{T_B}\ox M_{1,k}^{T_{B'}} - S_j^{T_{B}}\ox T_k^{T_{B'}}\|_{\infty}
    \leq & \|M_{0,j}^{T_B}\ox (M_{1,k}^{T_{B'}} - T_k^{T_{B'}})\|_{\infty}+ \|(M_{0,j}^{T_B}- S_j^{T_{B}})\ox T_k^{T_{B'}}\|_{\infty}\\
    {=} & \|M_{0,j}^{T_B}\|_{\infty}\|M_{1,k}^{T_{B'}} - T_k^{T_{B'}}\|_{\infty} + \|M_{0,j}^{T_B}- S_j^{T_{B}}\|_{\infty}\|T_k^{T_{B'}}\|_{\infty}\\
    \leq & \|M_{0,j}^{T_B}\|_{\infty} \mathscr{D}_{\infty}(\boldsymbol{M}_1) + \mathscr{D}_{\infty}(\boldsymbol{M}_0)
\end{aligned}
\end{equation}
Similarly, we have
\begin{equation}\label{Eq:add_ineq1}
    \|M_{0,j}^{T_B}\ox M_{1,k}^{T_{B'}} - S_j^{T_{B}}\ox T_k^{T_{B'}}\|_{\infty} \leq \|M_{1,k}^{T_B}\|_{\infty} \mathscr{D}_{\infty}(\boldsymbol{M}_0) + \mathscr{D}_{\infty}(\boldsymbol{M}_1).
\end{equation}
Combining Eq.~\eqref{Eq:add_ineq0} and Eq.~\eqref{Eq:add_ineq1}, we have
\begin{equation}
\begin{aligned}
    \|M_{0,j}^{T_B}\ox M_{1,k}^{T_{B'}} - S_j^{T_{B}}\ox T_k^{T_{B'}}\|_{\infty}
    \leq& \frac{1+\|M_{1,k}^{T_B}\|_{\infty}}{2} \mathscr{D}_{\infty}(\boldsymbol{M}_0) + \frac{1+\|M_{0,j}^{T_B}\|_{\infty}}{2}\mathscr{D}_{\infty}(\boldsymbol{M}_1).
\end{aligned}
\end{equation}
Therefore, we have
\begin{equation}
 \mathscr{D}_{\infty}(\boldsymbol{M}_0\ox\boldsymbol{M}_1) \leq \xi \cdot \mathscr{D}_{\infty}(\boldsymbol{M}_0) + \zeta \cdot \mathscr{D}_{\infty}(\boldsymbol{M}_1),
\end{equation}
where $\xi = \max_k (1+\|M_{1,k}^{T_B}\|_{\infty})/2$ and $\zeta = \max_j (1+\|M_{0,j}^{T_B}\|_{\infty})/2$.
\end{proof}

\subsection{Relative spectral PPT-distance}
In the following, we introduce the relative spectral PPT-distance of a POVM by considering $\Delta(M_j^{T_B}\|E_j^{T_B}) := \min \{r~| -rE_j^{T_B} \leq M_j^{T_B}-E_j^{T_B}\leq rE_j^{T_B} \}$, where $E_j^{T_B} \geq 0$ as $\{E_j\}_j \in \boldsymbol{\rm PPT}$. It can be checked that $\Delta(M_j^{T_B}\|E_j^{T_B}) = 0$ if and only if $M_j^{T_B}  = E_j^{T_B}$.
\begin{definition}[Relative spectral PPT-distance]
Given a POVM $\boldsymbol{M}=\{M_j\}_{j=0}^{n-1}$ acting on a bipartite system $\cH_A\ox\cH_B$, its relative spectral PPT-distance is defined as
\begin{equation}\label{eq:relative_spectral_ppt}
\mathscr{D}^{\cR}_{\infty}(\boldsymbol{M}) = \min\bigg\{\, r\,\bigg|\, -r E_j^{T_B} \leq M_j^{T_B} - E_j^{T_B} \leq r E_j^{T_B},\, \{E_j\}\in \boldsymbol{\rm PPT}_{AB}\bigg\},
\end{equation}
where the minimization ranges over all PPT POVMs.
\end{definition}
It is faithful as it can be checked that $\mathscr{D}^{\cR}_{\infty}(\boldsymbol{M})$ = 0 if and only if $\boldsymbol{M}$ is a PPT POVM, i.e. $\mathscr{D}^{\cR}_{\infty}(\boldsymbol{M}) = 0$ $\iff$ $\boldsymbol{M} \in \boldsymbol{\rm PPT}_{AB}$. The relative spectral PPT-distance also satisfies additional properties as shown below.
\begin{proposition}[Local unitary invariant]
    Given a POVM $\boldsymbol{M}=\{M_j\}_{j=0}^{n-1}$ acting on a bipartite system $\cH_A\ox\cH_B$, the relative spectral PPT-distance satisfies,
\begin{equation}
    \mathscr{D}^{\cR}_{\infty}((U_A\ox U_B)\boldsymbol{M}(U_A\ox U_B)^\dagger) = \mathscr{D}^{\cR}_{\infty}(\boldsymbol{M}),
\end{equation}
where $(U_A\ox U_B)\boldsymbol{M}(U_A\ox U_B)^\dagger$ denotes $\{(U_A\ox U_B)M_j(U_A\ox U_B)^\dagger\}_{j=0}^{n-1}$.
\end{proposition}
\begin{proof}
Let $\mathscr{D}^{\cR}_{\infty}(\boldsymbol{M}) = r_0$ with an optimal PPT POVM $\{E_j\}$. For any unitary operation $U_A\ox U_B$, denote $\boldsymbol{M'} = \{M_j'\}$, where each element $M_j'$ is given by $M_j' := (U_A\ox U_B)M_j(U_A\ox U_B)^\dagger)$. 
It is easy to see that $\{E_j':= (U_A\ox U_B) E_j(U_A\ox U_B)^\dag \}_j$ is a PPT POVM such that
\begin{equation}
    -r_0 E_j'^{T_B} \leq M_j'^{T_B} - E_j'^{T_B} \leq r_0 E_j'^{T_B}.
\end{equation}
It follows that $\mathscr{D}^{\cR}_{\infty}(\boldsymbol{M'}) \leq r_0$. Similarly, suppose $\mathscr{D}^{\cR}_{\infty}(\boldsymbol{M'}) = r_1$ with an optimal PPT POVM $\{F_j\}$. We have $\{F_j' := (U_A\ox U_B)^\dagger F_j (U_A\ox U_B)\}$ is a PPT POVM such that \begin{equation}
    -r_1 F_j'^{T_B} \leq M_j^{T_B} - F_j'^{T_B} \leq r_1 F_j'^{T_B},
\end{equation}
which gives $r_0 \le r_1$. Therefore, we have $\mathscr{D}^{\cR}_{\infty}(\boldsymbol{M'}) = \mathscr{D}^{\cR}_{\infty}(\boldsymbol{M})$.
\end{proof}

\begin{proposition}
For POVMs $\boldsymbol{M}_0, \boldsymbol{M}_1$ acting on a bipartite system $\cH_A\ox\cH_B$ and $0\leq \alpha \leq 1$,
\begin{equation}
    \mathscr{D}^{\cR}_{\infty}(\alpha\boldsymbol{M}_0 + (1-\alpha)\boldsymbol{M}_1) \leq \max\{\mathscr{D}^{\cR}_{\infty}(\boldsymbol{M}_0), \mathscr{D}^{\cR}_{\infty}(\boldsymbol{M}_1)\}.
\end{equation}
\end{proposition}
\begin{proof}
Let $\mathscr{D}^{\cR}_{\infty}(\boldsymbol{M}_0)= r_0, \mathscr{D}^{\cR}_{\infty}(\boldsymbol{M}_1) = r_1$ and denote the optimal POVMs for $\mathscr{D}^{\cR}_{\infty}(\boldsymbol{M}_0)$ and $\mathscr{D}^{\cR}_{\infty}(\boldsymbol{M}_1)$ are $\{S_{j}\}_j$ and $\{T_j\}_j$, respectively. By definition, we have $\forall j$
\begin{equation}
    -r_0 S_j^{T_B} \leq M_{0,j}^{T_B} - S_j^{T_B} \leq r_0 S_j^{T_B},
\end{equation}
and 
\begin{equation}
    -r_1 T_j^{T_B} \leq M_{1,j}^{T_B} - T_j^{T_B} \leq r_1 T_j^{T_B}.
\end{equation}
Since $S_j\geq 0$ and $T_j\geq 0$, it follows that 
\begin{equation}
\begin{aligned}
    &\; \alpha M_{0,j}^{T_B}+ (1-\alpha) M_{1,j}^{T_B}- \alpha S_j^{T_B} - (1-\alpha) T_j^{T_B} \\
    \leq&\; \alpha r_0 S_j^{T_B} + (1-\alpha)r_1 T_j^{T_B}\\
    \leq&\; \max\{r_0,r_1\}\left(\alpha S_j^{T_B} + (1-\alpha) T_j^{T_B}\right)
\end{aligned}
\end{equation}
Similarly, we have
\begin{equation}
\begin{aligned}
    &\; \alpha M_{0,j}^{T_B}+ (1-\alpha) M_{1,j}^{T_B}- \alpha S_j^{T_B} - (1-\alpha) T_j^{T_B} \\
    \geq&\; - \alpha r_0 S_j^{T_B} - (1-\alpha)r_1 T_j^{T_B}\\
    \geq&\; - \max\{r_0,r_1\}\left(\alpha S_j^{T_B} + (1-\alpha) T_j^{T_B}\right)
\end{aligned}
\end{equation}
We know that the convex combination $\{\alpha S_j + (1-\alpha)T_j\}_j$ is a valid PPT POVM. It follows that 
\begin{equation}
    \Delta\left( \alpha M_{0,j}^{T_B}+ (1-\alpha) M_{1,j}^{T_B} ~\big\|~ \alpha S_j^{T_B} - (1-\alpha) T_j^{T_B}\right) \leq \max\{r_0,r_1\},~\forall j
\end{equation}
Hence, we have $\mathscr{D}^{\cR}_{\infty}(\alpha\boldsymbol{M}_0 + (1-\alpha)\boldsymbol{M}_1) \leq \max\{\mathscr{D}^{\cR}_{\infty}(\boldsymbol{M}_0), \mathscr{D}^{\cR}_{\infty}(\boldsymbol{M}_1)\}$.
\end{proof}

Apart from these properties and its intuition for characterizing how much a POVM deviates from PPT POVMs, we will demonstrate that this quantity is directly related to the amount of entanglement required for optimal discrimination using PPT POVMs in the next section.

\section{Entanglement cost of PPT discrimination}
Now, we study the amount of entanglement required to discriminate a set of quantum states by PPT POVMs optimally. Herein, optimality refers to achieving the maximal average success probability without any restrictions on POVMs. We note the spectral PPT-distance of a POVM provides insights into the minimum amount of entanglement required to promote a PPT POVM to the given POVM. We will demonstrate this aspect through its application in quantum state discrimination tasks in the following.

Recall the optimal average success probability of discriminating two quantum states is given by Eq.~\eqref{Eq:suc_prob} when there are no restrictions on the POVM one is allowed to perform. However, an inherent challenge that arises in practice is the spatial distribution of quantum facilities, leading to the constraint that each laboratory can execute local quantum operations with mutual classical communication available. This motivates us to consider POVMs under locality constraints, including LOCC POVMs, SEP POVMs, and PPT POVMs as illustrated in Eq.~\eqref{Eq:povm_hire}. Despite extensive research characterizing the minimum amount of entanglement required to discriminate quantum states by POVMs under locality constraints~\cite{Yu2014,Bandyopadhyay2015LimSep, Gungor2016, Bandyopadhyay2018, Bandyopadhyay_2021}, a quantitative analysis is still lacking. To address this problem, we first introduce the entanglement-assisted average success probability where two parties are allowed to share entanglement to assist their QSD tasks.

\begin{figure}[t]
    \centering
    \includegraphics[width=0.55\linewidth]{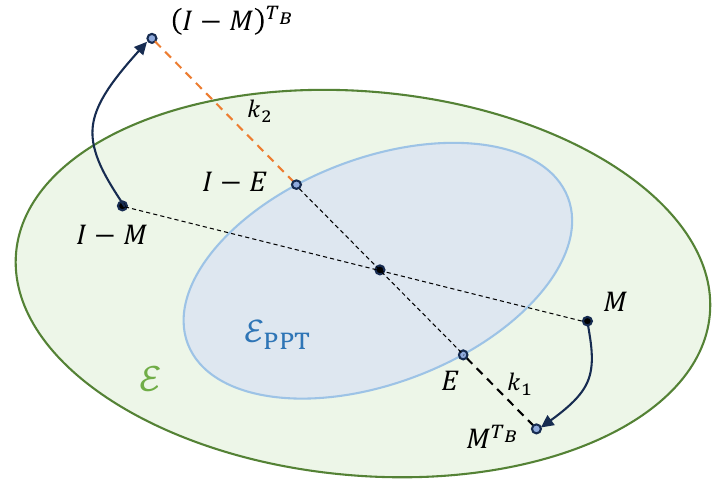}
    \caption{The spectral PPT-distance of a given binary-valued POVM $\boldsymbol{M} = \{M,I-M\}$. It is defined as the minimum spectral distance from the partial transpose of $M, I-M$ to the set of PPT POVMs. That is, $\mathscr{D}_{\infty}(\boldsymbol{M})=\min\max\{k_1,k_2\}$ where the minimization ranges over all possible choices of $E, I-E \in\cE_{\PPT}$. Each POVM is associated with a pair of points in the set, connected by a dashed line that intersects the origin point $I/2$.}
    \label{fig:Ent_cost_dis}
\end{figure}

\begin{definition}[Entanglement-assisted average success probability]
Given a quantum state ensemble $\Omega = \{(p_j,\rho_j)\}_{j=0}^{n-1}$, where $p_j\geq 0$ and $\sum_j p_j =1$, and a positive integer $k\geq 2$, the entanglement-assisted average success probability of discriminating $\Omega$ via operation class $\Pi$ is defined as
\begin{equation}
    P_{\suc,e}^{\Pi}(\Omega; k) = \max_{ \{E_j\} \in \Pi} \sum_{j=0}^{n-1} p_j\tr[E_j (\rho_j \ox \Phi_k^+)],
\end{equation}
where $\{E_j\}_{j=0}^{n-1}$ is a POVM,  $\Pi\in\{\LOCC,\SEP,\PPT\}$.
\end{definition}

We also denote $P_{\suc}^{\Pi}(\Omega)$ as the optimal average success probability when using POVMs in $\Pi$. Then it is natural to formally define the \textit{entanglement cost of optimal $\Pi$-discrimination} by considering how much entanglement it needs to render this entanglement-assisted average success probability identical to the original optimal success probability with no restrictions on POVMs.
\begin{definition}[Entanglement cost of optimal $\Pi$-discrimination]
Given a quantum state ensemble $\Omega = \{(p_j,\rho_j)\}_{j=0}^{n-1}$, where $p_j\geq 0$ and $\sum_j p_j =1$, the entanglement cost of optimal $\Pi$-discrimination is defined as
\begin{equation}
    E^{\Pi}_{C}(\Omega) = \min \left\{\log k:  P_{\suc,e}^{\Pi}(\Omega; k) = P_{\suc}(\Omega), k \geq 2\right\}.
\end{equation}
If $P_{\suc}^{\Pi}(\Omega) = P_{\suc}(\Omega)$, then $E^{\Pi}_{C}(\Omega)$ is set to be zero.
\end{definition}
Throughout the paper, we take the logarithm to be base two unless stated otherwise. In particular, consider the scenario of discriminating a pair of quantum states $\rho_0$ and $\rho_1$ given with the same prior probability as the results for asymmetric probabilities can be easily generalized. We write the entanglement-assisted average success probability of discriminating $\rho_0$ and $\rho_1$ by PPT POVMs as $P_{\suc,e}^{\PPT}(\rho_0, \rho_1; k)$ and establish its SDP form in the following Proposition~\ref{prop:ppt_ave_suc_prob}.
The derivation of its dual form can be found in Appendix~\ref{appendix:dual_EAASP}.

\begin{proposition}\label{prop:ppt_ave_suc_prob}
Given two bipartite quantum states $\rho_0,\rho_1 \in \cD(\cH_A\ox\cH_B)$ with the same prior probability, the entanglement-assisted average success probability of PPT-discrimination is
\begin{equation}
\begin{aligned}\label{Eq:ea_psuc}
P_{\suc,e}^{\PPT}(\rho_0, \rho_1 ; k) = &\max_{W_{j},Q_{j}}\; \frac{1}{2}\tr(\rho_0W_0 + \rho_1W_1) \\ 
    &\;\text{ s.t. }\left\{
     \begin{aligned}
     &\; W_j,Q_j \in \cE,\, j=0,1,\\
     &\; W_0+W_1=Q_0+Q_1=I,\\
     & -kQ_j^{T_B}\leq W_j^{T_B}-Q_j^{T_B} \leq  kQ_j^{T_B}.
     \end{aligned}
     \right.
\end{aligned}
\end{equation}
\end{proposition}
\begin{proof}
By definition, the entanglement-assisted average success probability by PPT POVMs can be computed via the following SDP,
\begin{align}\label{Eq:suc_prob_k}
P_{\suc,e}^{\PPT}(\rho_0, \rho_1; k) = \max &\; \frac{1}{2}+\frac{1}{2}\tr\Big[E [(\rho_0-\rho_1) \ox \Phi_k^{+}]\Big] \nonumber \\
     {\rm s.t.} & \;\; 0\leq E^{T_{BB'}}\leq I ,\\
     & \;\; 0\leq E \leq I \nonumber,
\end{align}
where $\Phi^{+}_k$ is the maximally entangled state on $\cH_{A'}\ox\cH_{B'}$ with $d_{A'}=d_{B'}=k$.
To further simplify the optimization problem, we notice that we can write 
\begin{equation}
\tr[E( (\rho_0 - \rho_1) \ox \Phi_k^{+}) ] = \tr[(\rho_0 - \rho_1)\tr_{A'B'}((I_{AB}\ox \Phi_k^{+})E)]. 
\end{equation}
Suppose $E$ is an optimal solution. Since $\Phi_k^{+}$ is invariant under any local unitary $U_{A'}\ox U^*_{B'}$, i.e., $(U_{A'}\ox U^*_{B'})\Phi_k^{+}(U_{A'}\ox U^*_{B'})^\dagger = \Phi_k^{+}$, $(U_{A'}\ox U^*_{B'})E(U_{A'}\ox U^*_{B'})^\dagger$ is also an optimal solution. Since any convex combination of optimal solutions remains optimal, we have that
\begin{equation}
    \widetilde{E} :=\int dU  (U_{A'}\ox U^*_{B'}) E (U_{A'}\ox U^*_{B'})^\dagger
\end{equation}
is also optimal, where $dU$ is the Haar measure. {The resulting operator will commute with all local unitaries $U_{A'}\ox U^*_{B'}$, as it is now invariant under their action. By applying Schur’s lemma~\cite{woit2017quantum} on the unitary group $\mathrm{U}(A)\ox \mathrm{U}(B)$, $E$ must act as a scalar multiple of the identity on each invariant subspace of $\cH_A\ox\cH_B$. This means that $E$ can be decomposed into a form where one component acts on the subspace spanned by $\Phi_k^{+}$ and the other on the orthogonal subspace.  
} Then we can restrict the optimal $\widetilde{E}$ as
\begin{equation}\label{eq:schur_optimal}
    \widetilde{E} = W_{AB}\ox\Phi_k^{+} + Q_{AB}\ox (I-\Phi_k^{+}),
\end{equation}
with certain linear operators $W_{AB}$ and $Q_{AB}$. It follows that
\begin{equation}
    \tr_{A'B'}[(I_{AB}\ox \Phi_k^{+})\widetilde{E}]= \tr_{A'B'}[ (I_{AB}\ox \Phi_k^{+})(W_{AB}\ox\Phi_k^{+} + Q_{AB}\ox (I-\Phi_k^{+})) ] = W_{AB},
\end{equation}
which can simplify the objective function as
\begin{equation}
    \tr\left[ {\widetilde{E}}[(\rho_0 - \rho_1) \otimes \Phi_k^+] \right] = \tr\left[ (\rho_0 - \rho_1) W_{AB}\right].
\end{equation}
By spectral decomposition, we have $(\Phi_k^{+})^{T_B}=\left(P_{+}-P_{-}\right)/k$, where $P_{+}$ and $P_{-}$ are the symmetric and anti-symmetric projections respectively. Then it follows that
\begin{equation}
\begin{aligned}
\widetilde{E}^{T_{B B^{\prime}}} &=W_{A B}^{T_B} \otimes (\Phi_k^{+})^{T_{B^{\prime}}}+Q_{A B}^{T_B} \otimes\left(I-\Phi_k^{+}\right)^{T_{B^{\prime}}} \\
& =W_{A B}^{T_B} \otimes \frac{P_{+}-P_{-}}{k}+Q_{A B}^{T_B} \otimes \frac{(k-1) P_{+}+(k+1) P_{-}}{k} \\
& =\left[W_{A B}^{T_B}+(k-1) Q_{A B}^{T_B}\right] \ox \frac{P_{+}}{k} +\left[-W_{A B}^{T_B}+(k+1) Q_{A B}^{T_B}\right] \ox \frac{P_{-}}{k}.
\end{aligned}
\end{equation}
Since $P_{+}$ and $P_{-}$ are positive and orthogonal to each other, we have $0 \leq \widetilde{E}^{T_{BB'}} \leq I$ if and only if 
\begin{equation}\label{eq:tildee_less_I}
\begin{aligned}
    & 0\leq W_{A B}^{T_B}+(k-1) Q_{A B}^{T_B} \leq kI_{AB}, \\
    & 0\leq -W_{A B}^{T_B}+(k+1) Q_{A B}^{T_B} \leq kI_{AB}.
\end{aligned}
\end{equation}
Similarly, we have $0 \leq \widetilde{E} \leq I$ if and only if
\begin{equation}\label{eq:e_less_I}
    0 \leq W_{AB} \leq I_{AB},\; 0 \leq Q_{AB} \leq I_{AB}
\end{equation}
{Combining Eq.~\eqref{eq:tildee_less_I} and~\eqref{eq:e_less_I}, the original optimization problem becomes,
\begin{equation}
\begin{aligned}\label{Eq:w0_ppt_optimization}
P_{\suc,e}^{\PPT}(\rho_0, \rho_1 ; k) = &\max_{W_{j},Q_{j}}\; \frac{1}{2} + \frac{1}{2}\tr[(\rho_0 - \rho_1) W_{AB}] \\ 
    &\;\text{ s.t. }\left\{
     \begin{aligned}
     &\; 0 \leq W_{AB} \leq I_{AB}, \;\; 0 \leq Q_{AB} \leq I_{AB}, \\
     &\; -(k-1) Q_{A B}^{T_B} \leq W_{A B}^{T_B} \leq  (1+k) Q_{A B}^{T_B} ,\\
     &\; -k (I_{AB} - Q_{A B}^{T_B}) + Q_{A B}^{T_B} \leq W_{A B}^{T_B} \leq k (I_{AB} - Q_{A B}^{T_B}) + Q_{AB}^{T_B}.
     \end{aligned}
     \right.
\end{aligned}
\end{equation}
Denoting $W_{AB} := W_0$, $Q_{AB} := Q_1$, $Q_1 := I - Q_{0}$ and $W_1 := I - W_{0}$, the last two constraints can be written as,
\begin{equation}\label{eq:w0_w1_ppt}
\begin{aligned}
    -kQ_{0}^{T_B} &\leq W_0^{T_B} - Q_0^{T_B} \leq kQ_{0}^{T_B}, \\
    -kQ_{1}^{T_B} &\leq W_1^{T_B} - Q_1^{T_B} \leq kQ_{1}^{T_B}.
\end{aligned}
\end{equation}
The objective function can then be equivalently written as,
\begin{equation}\label{eq:w0_w1_obj}
    \frac{1}{2} + \frac{1}{2}\tr[(\rho_0 - \rho_1)W_0] = \frac{1}{2}\tr[\rho_0 W_0 + \rho_1 W_1].
\end{equation}
Taking the constraints in Eq.~\eqref{eq:w0_w1_ppt} and objective function in Eq.~\eqref{eq:w0_w1_obj} into the optimization problem in Eq.~\eqref{Eq:w0_ppt_optimization}, the entanglement cost of optimal PPT-discrimination can be simplified and calculated by the optimization problem in Eq.~\eqref{Eq:ea_psuc}.}
\end{proof}

Building on the formulation of the entanglement-assisted average success probability in proposition~\ref{prop:ppt_ave_suc_prob}, we can show that the entanglement cost of optimal PPT-discrimination of $\rho_0, \rho_1$ is actually characterized by the minimum relative spectral PPT-distance of all possible unrestricted optimal POVMs.
\begin{proposition}\label{prop:ent_cost_dis}
Given two bipartite quantum states $\rho_0,\rho_1 \in \cD(\cH_A\ox\cH_B)$ with the same prior probability, the entanglement cost of optimal PPT-discrimination is given by,
\begin{equation}\label{Eq:ent_cost_dis}
\begin{aligned}
    E^{\PPT}_{C}(\rho_0,\rho_1) = \log &\min \lceil \mathscr{D}^{\cR}_{\infty}(\boldsymbol{W}) \rceil \\
    &\textrm{ s.t. } \left\{
    \begin{aligned}
        &\boldsymbol{W}=\{W_0,W_1\} \text{ is a POVM,}\\
        &\tr[(\rho_0 - \rho_1) W_0] = \frac{1}{2}\|\rho_0 - \rho_1\|_1.\\
    \end{aligned}
    \right.
\end{aligned}
\end{equation}
\end{proposition}
\begin{proof}
Based on Proposition~\ref{prop:ppt_ave_suc_prob}, the entanglement cost of optimal PPT-discrimination of $\rho_0, \rho_1$ can be written as the following optimization problem,
\begin{equation}\label{Eq:ent_cost}
\begin{aligned}
     E^{\PPT}_{C}(\rho_0,\rho_1) = \log &\min \;\lceil k\rceil \\
     &\text{ s.t. } \left\{ 
     \begin{aligned}
     &\; W_j,Q_j \in \cE,\, j=0,1,\\
     &\; W_0+W_1=Q_0+Q_1=I,\\
     &\;\tr\left[ (\rho_0 - \rho_1) W_0\right] = \frac{1}{2}\| \rho_0 - \rho_1\|_1,\\
     &-kQ_j^{T_B}\leq W_j^{T_B}-Q_j^{T_B} \leq  kQ_j^{T_B}.\\
     \end{aligned}
     \right.
\end{aligned}
\end{equation}
It can be noticed that the last constraint and the objective function are related to the relative spectral PPT-distance of $\boldsymbol{W}$ as
\begin{equation}
    \mathscr{D}^{\cR}_{\infty}(\boldsymbol{W}) = \min\bigg\{\, k\,\bigg|\, -k Q_j^{T_B} \leq W_j^{T_B} - Q_j^{T_B} \leq k Q_j^{T_B},\, \{Q_j\}\in \boldsymbol{\rm PPT}_{AB}\bigg\},
\end{equation}
By taking the ceiling of the objective function, it follows that the optimization for entanglement cost in Eq.~\eqref{Eq:ent_cost} can be equivalently expressed as in Eq.~\eqref{Eq:ent_cost_dis}.
\end{proof}

Proposition~\ref{prop:ent_cost_dis} reveals that an optimal general entangled POVM having the minimum relative spectral PPT-distance can be used to fully characterize the entanglement cost of optimal PPT-discrimination. 
\begin{remark}
We emphasize that the optimization problem for the entanglement cost in Eq.~\eqref{Eq:ent_cost_dis} is not a valid SDP. However, it reveals two essential aspects of the entanglement cost of optimal discrimination. Firstly, we derive an efficiently computable SDP lower bound for it which is provided in Appendix~\ref{Appendix:sdp_lower_upper}. Notably, by taking the ceiling of the solution to an SDP, the spectral PPT-distance of a POVM also lower bound the entanglement cost of optimal PPT-discrimination, i.e.,
\begin{equation}\label{Eq:lb_spec_dist}
\begin{aligned}
    E^{\PPT}_{C}(\rho_0,\rho_1) = \log \min \lceil \mathscr{D}^{\cR}_{\infty}(\boldsymbol{W}) \rceil \geq &\log \min \lceil \mathscr{D}_{\infty}(\boldsymbol{W})\rceil \\ 
    &\textrm{ s.t. } \left\{
    \begin{aligned}
        &\boldsymbol{W}=\{W_0,W_1\} \text{ is a POVM,}\\
        &\tr[(\rho_0 - \rho_1) W_0] = \frac{1}{2}\|\rho_0 - \rho_1\|_1.
    \end{aligned}
    \right.
\end{aligned}
\end{equation}
This can be checked by the definition of Eq.~\eqref{Eq:lb_spec_dist} and each feasible solution for the optimization problem Eq.~\eqref{Eq:ent_cost} is also a feasible solution to Eq.~\eqref{Eq:lb_spec_dist}. Secondly, we will further show in Theorem~\ref{thm:SDP_hierarchy} that the optimization problem in Eq.~\eqref{Eq:ent_cost_dis} can be computed via an SDP hierarchy that helps us precisely estimate the entanglement cost of optimal PPT-discrimination. This result will demonstrates that the relative spectral distance of a POVM is interlinked with the entanglement cost of optimal PPT-discrimination in entanglement-assisted quantum state discrimination tasks, thereby providing an operational meaning for this quantity. 
\end{remark}

Before showing the SDP hierarchy, we first present an upper bound on the entanglement cost of optimal PPT-discrimination for two arbitrary quantum states.

\begin{theorem}\label{thm:spec_upperbound}
Given two bipartite quantum states $\rho_0,\rho_1 \in \cD(\cH_A\ox\cH_B)$ with the same prior probability, the entanglement cost of optimal PPT-discrimination satisfies
\begin{equation}
    E_{C}^{\PPT}(\rho_0 , \rho_1) \leq \log \left\lceil \left\| M_+^{T_B} - M_-^{T_B} \right\|_\infty \right\rceil,
\end{equation}
where $M_{+}$ and $M_{-}$ are projections onto the positive and non-positive eigenspaces of $\rho_0 - \rho_1$, respectively.
\end{theorem}
We leave the proof in Appendix~\ref{appendix:proof_of_thms}. Notably, Theorem~\ref{thm:spec_upperbound} provides an easy-to-compute upper bound for the entanglement cost of optimal PPT-discrimination of a given pair of quantum states $\{\rho_0,\rho_1\}$ with equal prior probability, which only relies on the spectrum properties of the given states. Consequently, this upper bound can be directly determined by the ranks of $\rho_0$ and $\rho_1$ as the following corollary.

\begin{corollary}\label{cor:rank_Ec}
For two bipartite quantum states $\rho_0,\rho_1 \in \cD(\cH_A\ox\cH_B)$ with $\rank\rho_0 = r_0 \leq r_1 = \rank \rho_1$,
\begin{equation}
E_C^{\PPT} (\rho_0, \rho_1) \leq \log \eta,
\end{equation}
where $\eta = \max \{2r_0-1, r_0+1\}$.
\end{corollary}

Following Corollary~\ref{cor:rank_Ec}, we can now present an SDP hierarchy that accurately estimates the entanglement cost of PPT discrimination through successive iterations. This hierarchy is guaranteed to converge within a reasonable number of steps.
\begin{theorem}\label{thm:SDP_hierarchy}
    For two bipartite quantum states $\rho_0,\rho_1 \in \cD(\cH_A\ox\cH_B)$ with $\rank\rho_0 = r_0 \leq r_1 = \rank \rho_1$, there exists an SDP hierarchy for the entanglement cost of optimal $\PPT$-discrimination $E_{C}^{\PPT}(\rho_0, \rho_1)$ which converges in at most $\lceil \log(\eta)\rceil$ steps where $\eta = \max\{2r_0 - 1, r_0 +1 \}$.
\end{theorem}
\begin{proof}
Proposition~\ref{prop:ent_cost_dis} immediately implies a hierarchy of SDPs that can measure the disparity between the average success probability via global measurements and the entanglement-assisted average success probability of $\PPT$-discrimination as follows,
\begin{equation}
\begin{aligned}
     \Delta^{\PPT}(\rho_0, \rho_1; k) = \min\, & t \\
     \textrm{s.t. }& W_j,Q_j \in \cE,\, j=0,1,\\
     & W_0+W_1=Q_0+Q_1=I,\\
     & -t \leq \tr\left[ (\rho_0 - \rho_1) W_0\right] - \frac{1}{2}\| \rho_0 - \rho_1\|_1 \leq t,\\
     & -kQ_j^{T_B}\leq W_j^{T_B}-Q_j^{T_B} \leq  kQ_j^{T_B}, \forall j.\\
\end{aligned}
\end{equation}
where $k \in \mathbb{N}$ is fixed. It can be verified that this optimization problem is actually a valid SDP, since $k$ and $\frac{1}{2}\| \rho_0 - \rho_1\|_1$ are fixed scalars. It follows that $E_{C}^{\PPT}(\rho_0,\rho_1) = \log k_{\min}$, where $k_{\min}$ is the smallest integer $k$ such that $\Delta^{\PPT}(\rho_0,\rho_1 ; k) = 0$ as that is when the deviation of the entanglement-assisted average success probability of PPT-discrimination and average success probability of global POVMs is zero. Then, by Corollary~\ref{cor:rank_Ec}, the hierarchy converges in at most $\lceil\log(\eta) \rceil$ steps.
\end{proof}

\section{Examples}
{In this section, we explore the entanglement cost of PPT-discrimination for a pair of quantum states through several examples.}
We first consider that when two states are orthogonal mixed states, the optimal PPT-discrimination can be achieved with the assistance of $\log(d-1)$ ebits, which is less than standard quantum teleportation.
\begin{proposition}
For any two orthogonal $d$-dimensional mixed bipartite quantum states $\rho_0$ and $\rho_1$, 
\begin{equation}
    E_C^{\PPT}(\rho_0, \rho_1) \leq \log(d-1).
\end{equation}
\end{proposition}
\begin{proof}
    Denote the rank of $\rho_0$, $\rho_1$ as $\rank \rho_0 = r_0$ and $\rank\rho_1 = r_1$, respectively.
    Since $\rho_0$ and $\rho_1$ are mutually orthogonal, we have $r_0 + r_1 \leq d$.  Suppose that $r_0 \leq r_1$, we have that $\max r_0 = d/2$. According to Corollary~\ref{cor:rank_Ec}, we can conclude that $E_C^{\PPT}(\rho_0,\rho_1) \leq \log(d-1)$.
\end{proof}

In the following, we shall show that when one of the states in $\{\rho_0,\rho_1\}$ is pure, a single Bell state suffices to fully unlock the quantum data hiding against PPT POVMs, i.e., the optimal PPT-discrimination can always be achieved with the assistance of 1 ebit.
\begin{proposition}\label{prop:pure_mixed_Ec}
    For bipartite quantum states $\rho_0$ and $\rho_1$, if one of them is a pure state, then 
    \begin{equation}
        E_C^{\PPT} (\rho_0, \rho_1) \leq 1.
    \end{equation}
\end{proposition}
\begin{proof}
{By directly applying Corollary~\ref{cor:rank_Ec}, we have the optimal discrimination can be achieved with $E_C^{\PPT}(\rho_0, \rho_1) \leq 1$.}
\end{proof}

Proposition~\ref{prop:pure_mixed_Ec} establishes that any $d$-dimensional pure state can be optimally discriminated against any other state with the assistance of one Bell state. This result recovers and extends the previous result that one Bell state is enough for discriminating any $d$-dimensional pure state and its orthogonal complement~\cite{Yu2014}.

\begin{remark}
    We remark there is an intriguing connection between the entanglement cost of optimal PPT-discrimination and quantum data hiding~\cite{DiVincenzo_2002}. Quantum data hiding demonstrates the existence of bipartite state pairs optimally distinguishable by general entangled POVMs but not by locally constrained sets of measurements, e.g., LOCC, SEP, and PPT POVMs. This underscores a discernible gap in the efficacy of message decoding when encoded in pairs of states and the receiver is constrained to limited classes of POVMs. These can be understood by the data-hiding procedures depending on the given physical setting, e.g., two quantum devices for decoding information cannot generate reliable entanglement. For quantum data hiding associated with PPT POVMs, the data hiding ratio~\cite{Lami_2018} against PPT POVMs is given by,
    \begin{align}
        R(\PPT) = \max_{\{\rho_0, \rho_1, p\}} \frac{\| p\rho_0-(1-p)\rho_1 \|_{\text{ALL}}}{ \| p\rho_0-(1-p)\rho_1 \|_{\PPT}},
    \end{align}
    which quantifies the disparity between the capabilities of global POVMs and PPT POVMs in QSD. We showcase that the strategic use of maximally entangled states in QSD tasks serves as a key to unveiling these distinctions in message decoding capabilities between general entangled measurements and PPT POVMs. In particular, 
    Ref.~\cite{Matthews_2009} establishes a lower bound for $R(\PPT)$ with $R(\PPT) \geq (d+1)/2$ considering a bipartite system $\cH_d \otimes \cH_d$.
\end{remark}

Intriguingly, we observe that a single Bell state is not always enough for optimally discriminating an arbitrary pair of quantum states by LOCC operations or PPT POVMs.
\begin{proposition}\label{prop:LOCC_eg}
    There exists a pair of bipartite states $\{\rho, \sigma\}$ for which the optimal discrimination concerning general entangled POVMs is not achievable by LOCC with the assistance of a single Bell state.
\end{proposition}

This phenomenon is illustrated through a numerical example provided in Appendix~\ref{appendix:qutr_eg} and the numerical codes for estimating the required entanglement bits for quantum state discrimination can be found in~\cite{ourcode}, wherein we showcase a pair of two-qutrit quantum states that, even with the aid of a single Bell state, cannot be optimally discriminated by PPT POVMs.
While one Bell state is sufficient for any pure state and its orthogonal complement, it is not a universally sufficient solution for all state pairs. This example underscores a more sophisticated gap between the optimal general entangled POVMs and LOCC operations in the context of quantum state discrimination for mixed quantum states. 

\begin{figure}[t]
    \centering
    \subfloat{{\includegraphics[width=0.498\linewidth]{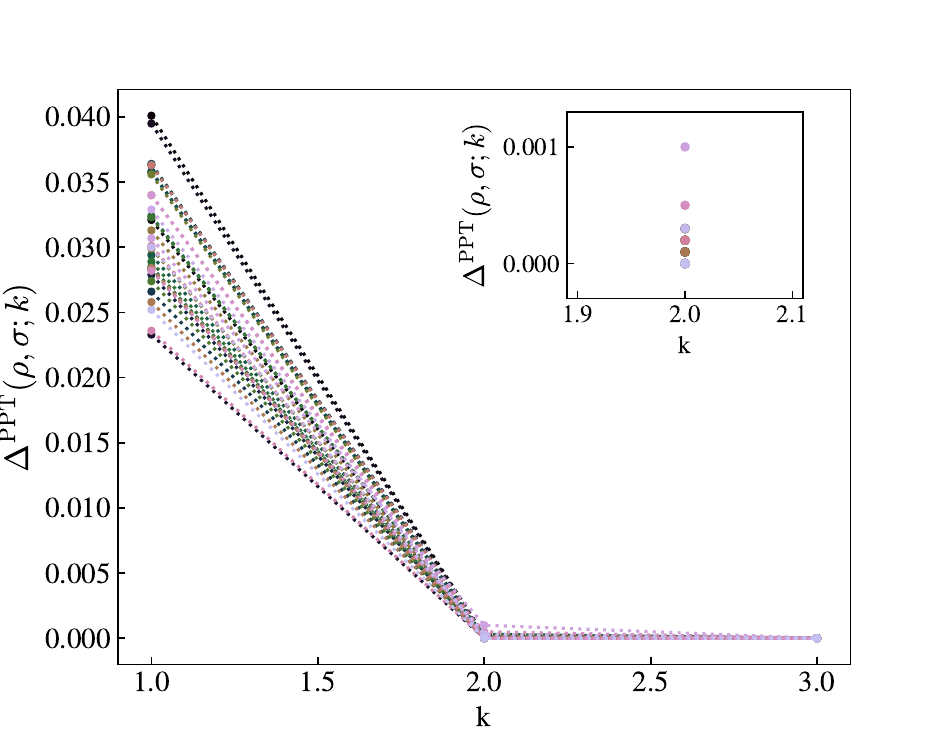} }}%
    \subfloat{{\includegraphics[width=0.498\linewidth]{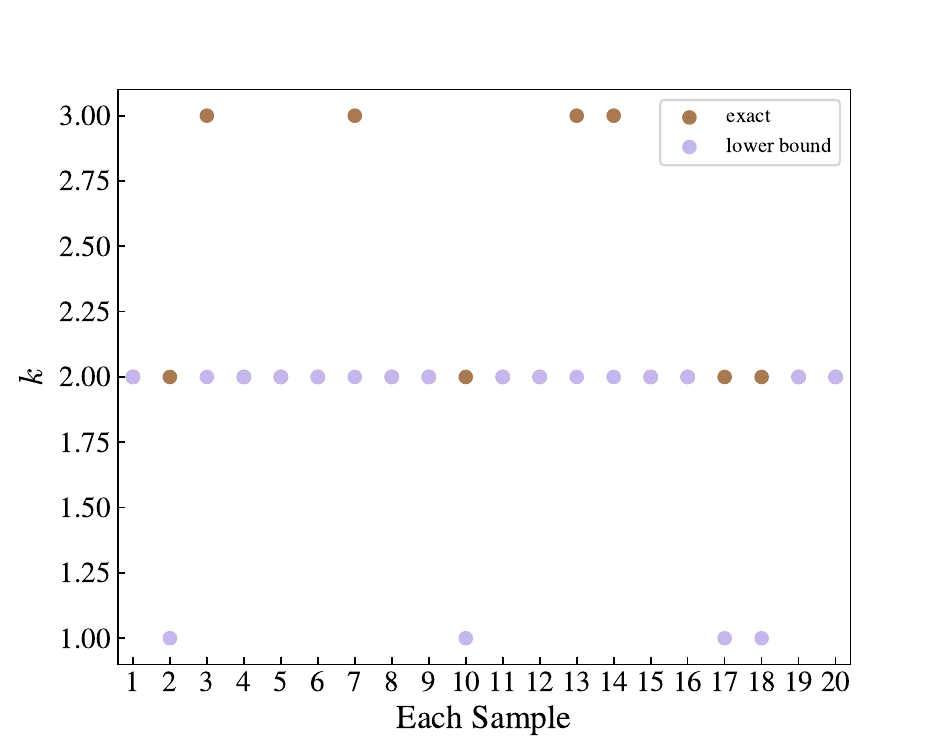} }}%
    \caption{(Left) Deviation between the entanglement-assisted average success probability of PPT discrimination and the average success probability of global POVMs for randomly sampled quantum states. The zoomed-in section illustrates that one ebit is not always sufficient for discriminating between two arbitrary mixed states, as evidenced by the disparity between the average success probabilities achieved through PPT POVMs and those achieved through global POVMs. (Right) Comparison of lower bound and exact value determined by SDP Hierarchy for entanglement cost in optimal PPT-discrimination. The brown points indicate the exact values and the purple points represent the lower bounds. It can show that the lower bound is tight in most cases.}%
    \label{fig:sdp_hierarchy_varies}%
\end{figure}

{Furthermore, we numerically tested the SDP hierarchy and the tightness of the lower bound of the entanglement cost of optimal PPT-discrimination. We first present additional numerical results to evaluate the effectiveness of this hierarchy in Theorem~\ref{thm:SDP_hierarchy} and examine the tightness of the upper bounds provided. Specifically, we analyze 35 pairs of qudit quantum states $\rho, \sigma \in \cD(\cH_4 \otimes \cH_4)$. $\rho$ is sampled as rank-3 states using Hilbert-Schmidt measures~\cite{zyczkowski2011generating}, while the rank of $\sigma$ is sampled randomly. For each pair of quantum states, we run the SDP hierarchy with fixed levels $k \in [1,2,3]$ and obtain the corresponding deviation between the entanglement-assisted average success probability of PPT discrimination and the average success probability of global POVMs. At $k=1$, the setting represents scenarios without entanglement assistance. We plot the variations in the entanglement-assisted average success probability of PPT discrimination across different $k$ values in Figure~\ref{fig:sdp_hierarchy_varies}. The results indicate that a single Bell state suffices for most pairs of quantum states, although there are instances where $k=3$ is necessary. These findings are consistent with the theoretical results from Theorem~\ref{thm:SDP_hierarchy} that the SDP hierarchy converges within at most $\lceil\eta\rceil$ steps, demonstrating alignment between our numerical simulations and the theoretical framework.}

{
Furthermore, Proposition~\ref{prop:sdp_low_bound_appendix} establishes a lower bound, which we demonstrate to be tight in several instances. Specifically, we have randomly sampled an additional 20 pairs of quantum states to illustrate the tightness of this lower bound. This sampling follows the same procedure. The results are displayed in Figure~\ref{fig:sdp_hierarchy_varies}, where the brown points indicate the exact values and the purple points represent the lower bounds. These numerical experiments suggest that the lower bound is indeed tight in most cases.
}

\section{Concluding Remarks}
In this work, we delve into quantifying the entanglement cost for optimally discriminating bipartite quantum states under locality constraints. 
{We introduce the spectral PPT-distance and the relative spectral PPT-distance of a POVM and relate these quantities to the operational task of determining the entanglement cost of optimal discrimination by PPT POVMs. Furthermore, we provide an easy-to-compute upper bound based on the spectrum properties of the given states and apply it to establish an SDP hierarchy for determining the precise amount of entanglement necessary to fully unlock quantum data hiding across various pairs of quantum states.} Specifically, our findings notably emphasize that $\log(d-1)$ ebits suffice to unlock the quantum data hiding against PPT POVMs on orthogonal mixed quantum states. This result shows the efficiency of ebits for discriminating quantum states compared to teleportation. 
We also extend the previous results on discriminating a pure state against another one and show that if the data encoding involves a pure state, a single Bell state is always sufficient. For two mixed states, we present an example of a pair of states that cannot be optimally discriminated via LOCC with the assistance of one Bell state.

An intriguing direction for future research lies in further analysis of quantum resource manipulation for state discrimination, e.g., within the framework of resource theories of coherence and magic states. Specifically, the study of constrained families of POVMs, such as incoherent POVMs~\cite{Oszmaniec_2019} and those with positive Wigner functions~\cite{zhu2023limitations}, presents a promising avenue. Additionally, the exploration of resource manipulation in quantum state discrimination, with respect to thermodynamics~\cite{Gour2015nonequilibrium, chiribella2022nonequilibrium} and imaginarity~\cite{Wu2021Imaginarity}, could pave the way for novel insights into the interplay between quantum resources and measurement processes.

\section*{Acknowledgement}
{We thank anonymous reviewers for their helpful comments and for suggesting the SDP hierarchy of calculating the entanglement cost of PPT-discrimination.} We would like to thank Hongshun Yao, Benchi Zhao, and Kun Wang for their helpful comments. Chenghong Zhu and Chengkai Zhu contributed equally to this work. 
This work was partially supported by the National Key R\&D Program of China (Grant No.~2024YFE0102500), the Guangdong Provincial Quantum Science Strategic Initiative (Grant No.~GDZX2403008, GDZX2403001), the Guangdong Provincial Key Lab of Integrated Communication, Sensing and Computation for Ubiquitous Internet of Things (Grant No.~2023B1212010007), the Quantum Science Center of Guangdong-Hong Kong-Macao Greater Bay Area, and the Education Bureau of Guangzhou Municipality.

\bibliographystyle{IEEEtran}
\bibliography{arxiv/arxiv}

\appendix

\setcounter{subsection}{0}
\setcounter{table}{0}
\setcounter{figure}{0}

\renewcommand{\theequation}{S\arabic{equation}}
\numberwithin{equation}{section}
\renewcommand{\theproposition}{S\arabic{proposition}}
\renewcommand{\thedefinition}{S\arabic{definition}}
\renewcommand{\thefigure}{S\arabic{figure}}
\setcounter{equation}{0}
\setcounter{table}{0}
\setcounter{section}{0}
\setcounter{proposition}{0}
\setcounter{definition}{0}
\setcounter{figure}{0}

\section{Dual of Entanglement-Assisted Average Success Probability}\label{appendix:dual_EAASP}
For simplicity, we ignore the system symbols, i.e. $W_{AB} = W$. By introducing the Lagrange multiplier $A$, $B$, $C$, $D$, $F$,$G$, the Lagrange function of the primal problem is
\begin{equation}
\begin{aligned}
    &L(A, B, C, D, F, G) \\
    &= \frac{1}{2} + \frac{1}{2}\tr[(\rho_0 - \rho_1)W] - \langle A, W-I\rangle - \langle B, Q-I\rangle \\
    &\quad - \langle C, (1-k)Q^{T_B} - W^{T_B}\rangle -\langle D, W^{T_B} - (1+k)Q^{T_B}\rangle \\ 
    &\quad - \langle F, (1+k)Q^{T_B} - kI - W^{T_B} \rangle - \langle G, W^{T_B} - kI + (k-1)Q^{T_B} \rangle \\
    &= \frac{1}{2} + \frac{1}{2}\tr[(\rho_0 - \rho_1)W] - \langle A, W\rangle + \langle A\rangle - \langle B, Q\rangle + \langle B\rangle \\ 
    & \quad -\langle C, (1-k)Q^{T_B}\rangle+ \langle C, W^{T_B}\rangle -\langle D, W^{T_B}\rangle + \langle D, (1+k)Q^{T_B}\rangle \\ 
    &\quad - \langle F, (1+k)Q^{T_B} \rangle + \langle F, kI\rangle + \langle F, W^{T_B} \rangle - \langle G, W^{T_B}\rangle + \langle G, kI \rangle - \langle G, (k-1)Q^{T_B} \rangle \\
    &= \frac{1}{2} + \tr[A] + \tr[B] + \tr[kF] + \tr[kG] \\
    &\quad + \frac{1}{2}\langle(\rho_0 - \rho_1),W\rangle - \langle A,W\rangle +\langle C^{T_B}, W\rangle - \langle D^{T_B}, W\rangle  + \langle F^{T_B}, W\rangle - \langle G^{T_B}, W\rangle \\
    &\quad - \langle B, Q\rangle - \langle (1-k)C^{T_B}, Q\rangle + \langle (1+k)D^{T_B}, Q\rangle - \langle (1+k)F^{T_B}, Q\rangle - \langle (k-1)G^{T_B}, Q\rangle. \notag
\end{aligned}
\end{equation}
The corresponding Lagrange dual function is
\begin{equation}
    g(A, B, C, D, F, G) = \inf_{W\geq 0,Q\geq 0} L(W, Q, A, B, C, D, F, G).
\end{equation}
Since $W \geq 0$ and $Q \geq 0$, it must hold that
\begin{subequations}
\begin{align}
    & \frac{1}{2}(\rho_0 - \rho_1) -A + C^{T_B} - D^{T_B} + F^{T_B} - G^{T_B} \leq 0 \\
    & -B - (1-k)C^{T_B} + (1+k)D^{T_B} - (1+k)F^{T_B} - (k-1)G^{T_B} \leq 0.
\end{align}
\end{subequations}
Thus the dual SDP is
\begin{align}
P_{\suc,e}^{\PPT}(\rho_0, \rho_1; k) = \min &\;\frac{1}{2} + \tr[A+B] + k\tr[F+G] \notag \\
     {\rm s.t.} & \; \frac{1}{2}(\rho_0 - \rho_1) -A + C^{T_B} - D^{T_B} + F^{T_B} - G^{T_B} \leq 0 \notag \\
     &  -B - (1-k)C^{T_B} + (1+k)D^{T_B}- (1+k)F^{T_B} - (k-1)G^{T_B} \leq 0.
\end{align}

\section{Proof of Theorems}\label{appendix:proof_of_thms}
\begin{lemma}[Weyl's inequality~\cite{horn2012matrix}]\label{lem:weyl_ineq}
Let $A, B \in M_n$ be Hermitian and let the respective eigenvalues of $A, B$, and $A+B$ be $\left\{\lambda_j(A)\right\}_{j=1}^n,\left\{\lambda_j(B)\right\}_{j=1}^n$, and $\left\{\lambda_j(A+B)\right\}_{j=1}^n$, in a non-decreasing order. Then
\begin{equation}\label{ieq:Weyl’sinequalities1}
    \lambda_i(A+B) \leq \lambda_{i+j}(A)+\lambda_{n-j}(B), \quad j=0,1, \ldots, n-i
\end{equation}
for each $i=1, \ldots, n$, with equality for some pair $i, j$ if and only if there is a nonzero vector $x$ such that $A x=\lambda_{i+j}(A) x, B x=\lambda_{n-j}(B) x$, and $(A+B) x=\lambda_i(A+B) x$. Also,
\begin{equation}\label{ieq:Weyl’sinequalities2}
\lambda_{i-j+1}(A)+\lambda_j(B) \leq \lambda_i(A+B), \quad j=1, \ldots, i    
\end{equation}
for each $i=1, \ldots, n$, with equality for some pair $i, j$ if and only if there is a nonzero vector $x$ such that $A x=\lambda_{i-j+1}(A) x, B x=\lambda_j(B) x$, and $(A+B) x=\lambda_i(A+B) x$. 
\end{lemma}

\begin{lemma}\label{lem:rank_n_projector}
For quantum states $\rho_0$ and $\rho_1$ with {$\rank(\rho_0) = r_0 \leq r_1 = \rank(\rho_1)$}, let $M_+$ denote the projector onto the positive eigenspace of $\rho_0 - \rho_1$. It satisfies that {$\rank(M_+) \leq r_0$}.
\end{lemma}
\begin{proof}
Denote $\rho_0 = \sum_j q_j \ketbra{\psi_j}{\psi_j}, \rho_1 = \sum_j p_j \ketbra{\phi_j}{\phi_j}$, where 
\begin{equation}
\begin{aligned}
    &q_1 = \cdots = q_{n-r_0} = 0 \leq q_{n-r_0+1} \leq \cdots \leq q_n < 1,\\
    &p_1 \geq \cdots \geq p_{r_1} \geq p_{r_1+1} = \cdots = p_n = 0.
\end{aligned}
\end{equation}
Consider the eigenvalues of $\rho_0-\rho_1$ in a non-decreasing order as $\{\lambda_k(\rho_0-\rho_1)\}_k$. 
By Eq.~\eqref{ieq:Weyl’sinequalities1} in Lemma~\ref{lem:weyl_ineq}, we have
    \begin{equation}
        \lambda_{n-r_0}(\rho_0-\rho_1) \leq q_{n-r_0} - p_n \leq 0.
    \end{equation}
    Then we can conclude that the number of positive eigenvalues of $\rho_0 - \rho_1$ is less than $r_0$. Thus, the projector $M_+$ onto the positive eigenspace of $\rho_0-\rho_1$ satisfies $\rank(M_+) \leq r_0$.
\end{proof}

\renewcommand\theproposition{\ref{thm:spec_upperbound}}
\setcounter{proposition}{\arabic{proposition}-1}
\begin{theorem}
Given two bipartite quantum states $\rho_0, \rho_1$ with the same prior probability, the entanglement cost of optimal PPT-discrimination satisfies
\begin{equation}
    E_{C}^{\PPT}(\rho_0 , \rho_1) \leq \log \left\lceil \left\| M_+^{T_B} - M_-^{T_B} \right\|_\infty \right\rceil,
\end{equation}
where $M_{+}$ and $M_{-}$ be the projectors onto the positive and non-positive eigenspaces of $\rho_0 - \rho_1$, respectively.
\end{theorem}
\begin{proof}
Given the quantum states $\rho_0$ and $\rho_1$, let $M_+$ be the projector onto the positive eigenspace of $\rho_0-\rho_1$.  We will show that if
\begin{equation}\label{eq:partial_tb_helstrom}
    \frac{1-k}{2} I\leq M_+^{T_B} \leq \frac{1+k}{2} I    
\end{equation}
satisfies for some positive integer $k$, $\{ W_{0} = M_+, Q_{0} = I/2, k\}$ is a feasible solution to the optimization problem in Eq.~\eqref{Eq:ent_cost}.
Taking $Q_{0} = I/2$ and $W_{0} = M_+$ and $(1-k)/2 \cdot I\leq M_+^{T_B} \leq (1+k)/2 \cdot I$, we can easily check that constraints in Eq.~\eqref{Eq:ent_cost} are satisfied.
Denoting $M_-$ as the projector onto the non-positive eigenspace of $\rho_0 - \rho_1$ such that $M_+ + M_- = I$, we arrange Eq.~\eqref{eq:partial_tb_helstrom} as $-kI \leq M_+^{T_B} - M_-^{T_B} \leq kI$ and complete the proof.
\end{proof}
\renewcommand{\theproposition}{S\arabic{proposition}}

\renewcommand\theproposition{\ref{cor:rank_Ec}}
\setcounter{proposition}{\arabic{proposition}-1}
\begin{corollary}
For two bipartite quantum states $\rho_0,\rho_1$ with $\rank\rho_0 = r_0 \leq r_1 = \rank \rho_1$,
\begin{equation}
E_C^{\PPT} (\rho_0, \rho_1) \leq \log \eta,
\end{equation}
where $\eta = \max \{2r_0-1, r_0+1\}$.
\end{corollary}
\begin{proof}
By Lemma~\ref{lem:rank_n_projector}, we have that the projector $M_+$ onto the positive eigenspace of $\rho_0 - \rho_1$ satisfies $\rank M_+ \leq r_0 $.
Suppose $M_+$ has a spectral decomposition $M_+ = \sum_j^{r_0} \ketbra{\psi_j}{\psi_j}$. It follows
\begin{equation}\label{Eq:min_eigen_PT}
    \lambda_{\min}(M_+^{T_B}) \geq \sum_j^{r_0} \lambda(\ketbra{\psi_j}{\psi_j}^{T_B}) \geq \sum_{j}^{r_0} -\frac{1}{2} = -\frac{1}{2}r_0,
\end{equation}
where the second inequality based on the result that the partial transpose of any pure state has eigenvalues between $\left[-1/2, 1\right]$~\cite{Rana2013}, i.e. $-\frac{1}{2} \leq \lambda_{\min}(\ketbra{\psi_j}{\psi_j}^{T_B}) \leq \lambda_{\max}(\ketbra{\psi_j}{\psi_j}^{T_B}) \leq 1$. Similarly, for $\lambda_{\max}$
\begin{equation}\label{Eq:max_eigen_PT}
    \lambda_{\max}(M_+^{T_B}) \leq \sum_j^{r_0} \lambda(\ketbra{\psi_j}{\psi_j}^{T_B}) \leq \sum_j^{r_0} 1 = r_0.
\end{equation}
Combining inequalities in Eq.~\eqref{Eq:min_eigen_PT} and Eq.~\eqref{Eq:max_eigen_PT}, we can deduce that
\begin{equation}
    - \frac{r_0}{2}I \leq M_+^{T_B} \leq r_0I.
\end{equation}
By Theorem~\ref{thm:spec_upperbound}, we have $E_C^{\PPT} (\rho_0, \rho_1) \leq \log \eta$, where $\eta = \max \{2r_0-1, r_0+1\}$,
\end{proof}
\renewcommand{\theproposition}{S\arabic{proposition}}

\section{Entanglement Cost of PPT discrimination SDP Lower and Upper Bounds}\label{Appendix:sdp_lower_upper}
\begin{proposition}\label{prop:sdp_low_bound_appendix}
Given quantum states $\rho_0$ and $\rho_1$, it satisfies that $E^{\PPT}_{C}(\rho_0,\rho_1) \geq \log \lceil \eta_l  \rceil$ where
\begin{align}
    \eta_l = \min &\;  k \notag \\
         {\rm s.t.} &\; 0\leq W_{AB}\leq I_{AB}, 0\leq \hat{Q}_{AB}\leq (k+1) I_{AB}, \notag\\
         & \;2\tr\left[ (\rho_1 - \rho_0) W_{AB}\right] = \| \rho_0 - \rho_1\|_1, \notag\\
         &  -\hat{Q}_{AB}^{T_B} \leq W_{A B}^{T_B} \leq  \hat{Q}_{AB}^{T_B} , \\
         & \;\hat{Q}_{AB}^{T_B} -k I_{AB} \leq W_{A B}^{T_B} \leq kI_{AB}. \notag
    \end{align}
\end{proposition}
\begin{proof}
    We first find the SDP lower bound for the entanglement cost of PPT discrimination. The basic idea is to relax the constraints of the bi-linear optimization problem and introduce a new variable.  Since $k \geq 0$ and $Q_{AB}^{T_B} \geq 0$, the constraint $(-k+1) Q_{A B}^{T_B} \leq W_{A B}^{T_B} \leq  (k+1) Q_{A B}^{T_B}$ can be relaxed to
    \begin{equation}
        -(k+1) Q_{A B}^{T_B} \leq W_{A B}^{T_B} \leq  (k+1) Q_{A B}^{T_B}.
    \end{equation}
    And the constraint $(k+1) Q_{A B}^{T_B} -k I_{AB} \leq W_{A B}^{T_B} \leq -(k-1)Q_{AB}^{T_B} + kI_{AB}$ can be relax to
    \begin{equation}
        (k+1)Q_{AB}^{T_B} -k I_{AB} \leq W_{A B}^{T_B} \leq kI_{AB}.
    \end{equation}
    By introducing a new variable $\hat{Q}_{AB} = (k+1)Q_{AB}$, we have the SDP lower bound.
\end{proof}

\begin{proposition}
Given quantum states $\rho_0$ and $\rho_1$, it satisfies that $E^{\PPT}_{C} (\rho_0, \rho_1)  \leq   \log \lceil \eta_u  \rceil$, where ,
    \begin{subequations}
    \begin{align}
    \eta_u = \log \min& \;  k \notag\\
        {\rm s.t.} & \; 0\leq W_{AB}\leq I_{AB}, 0\leq \hat{Q}_{AB}\leq (k+1) I_{AB}, \notag\\
         & \;2\tr\left[ (\rho_1 - \rho_0) W_{AB}\right] = \| \rho_0 - \rho_1\|_1 , \notag \\
         & \;0 \leq W_{A B}^{T_B} \leq \hat{Q}_{AB}^{T_B},\\
         & \;\hat{Q}_{AB}^{T_B} -k I_{AB} \leq W_{A B}^{T_B} \leq -\hat{Q}_{AB}^{T_B} + kI_{AB}. \notag
    \end{align}
    \end{subequations}
\end{proposition}
\begin{proof}
    We then find the SDP upper bound for the entanglement cost of PPT discrimination.  The key idea is to restrict the constraints of the bi-linear optimization problem.  Since $k \geq 0$ and $Q_{AB}^{T_B} \geq 0$, we can restrict the constraint $(-k+1) Q_{A B}^{T_B} \leq W_{A B}^{T_B} \leq  (k+1) Q_{A B}^{T_B}$ as
    \begin{equation}
        0 \leq W_{A B}^{T_B} \leq  (k+1) Q_{A B}^{T_B},
    \end{equation}
    where this restriction requires $k \geq 1$.
    And the constraint $(k+1) Q_{A B}^{T_B} -k I_{AB} \leq W_{A B}^{T_B} \leq - (k-1) Q_{A B}^{T_B} + kI_{AB}$ can be restricted to
    \begin{equation}
        (k+1) Q_{A B}^{T_B} -k I_{AB} \leq W_{A B}^{T_B} \leq - (k+1) Q_{A B}^{T_B} + kI_{AB}.
    \end{equation}
    Introducing a new variable $\hat{Q}_{AB} = (k+1)Q_{AB}$, we have the SDP upper bound.
\end{proof}


\section{Qutrit States Example}\label{appendix:qutr_eg}

\begin{figure}[H]
\begin{equation}\label{eq:qutrit_sigma}
    \sigma = \left[ 
    \begin{array}{*{9}{r}}
    0.0601  &  0.0327 &  -0.0601  & -0.0068  & -0.0024 &  -0.0040  & -0.0665  & -0.0261  & -0.0286 \\
    0.0327  &  0.0990 &  -0.0789  &  0.0068  & -0.0522  & -0.0660  & -0.0272  &  0.0309  & -0.0397 \\
   -0.0601 &  -0.0789  &  0.1812 &  -0.0067 &   0.0497 &   0.0254  &  0.0548  &  0.0620 &   0.0628\\
   -0.0068  &  0.0068  & -0.0067  &  0.0141  &  0.0108  & -0.0137  &  0.0229  &  0.0277  &  0.0164 \\
   -0.0024  & -0.0522 &   0.0497 &   0.0108 &   0.1078  & -0.0278 & -0.0031 &  -0.0218  &  0.0075 \\
   -0.0040 &  -0.0660  &  0.0254  & -0.0137  & -0.0278   & 0.1253  &  0.0173  & -0.0002  &  0.0593 \\
   -0.0665  & -0.0272  &  0.0548  &  0.0229 &  -0.0031 &   0.0173  &  0.1195 &   0.0973 &   0.0587 \\
   -0.0261 &   0.0309 &   0.0620 &   0.0277 &  -0.0218 &  -0.0002 &   0.0973  &  0.1906 &   0.0840 \\
   -0.0286 &  -0.0397 &   0.0628 &   0.0164 &   0.0075&    0.0593 &   0.0587 &   0.0840 &   0.1025 
    \end{array}    
    \right]
\end{equation}
\end{figure}

\renewcommand\theproposition{\ref{prop:LOCC_eg}}
\setcounter{proposition}{\arabic{proposition}-1}
\begin{proposition}
    There exists a pair of bipartite states $\{\rho, \sigma\}$ for which the optimal discrimination concerning general entangled POVMs is not achievable by LOCC with the assistance of a single Bell state.
\end{proposition}

\begin{proof}
Consider the following bipartite qutrit state $\rho$,
\begin{equation}\label{eq:qutrit_rho}
\begin{aligned}
    \ket{\psi_0} = a_0 \ket{01} + &a_1 \ket{02} + a_2 \ket{20}, \;\;  \ket{\psi_1} = a_0\ket{10} + a_1\ket{21} + a_2\ket{12}, \\
    & \rho = \frac{1}{5}\ketbra{\psi_0}{\psi_0} + \frac{4}{5}\ketbra{\psi_1}{\psi_1},
\end{aligned}
\end{equation}
where $a_0 = \frac{\sqrt{5}}{5}$, $a_1 = \frac{3\sqrt{5}}{10}$ and $a_2 =  \frac{1}{2}\sqrt{\frac{7}{5}}$, and a bipartite state $\sigma$ as shown in Eq.~\eqref{eq:qutrit_sigma}. For $\rho$ and $\sigma$, we computed the dual SDP in Eq.~\eqref{Eq:suc_prob_k} and primal SDP in Eq.~\eqref{Eq:suc_prob} using CVX~\cite{cvx,gb08} in Matlab~\cite{MATLAB}. Then borrowing the idea of \textit{computer-assisted proofs} given in Ref.~\cite{Bavaresco_2021}, we construct strictly feasible solutions for SDPs associated with $\{\rho,\sigma\}$ and $\{\rho\ox\Phi_2^+, \sigma\ox\Phi_2^+\}$. Consequently, we obtain
\begin{equation}
    P_{\suc,e}^{\PPT}(\rho,\sigma;2) \leq 0.9125 < 0.9135 \leq P_{\suc}(\rho,\sigma;2),
\end{equation}
where the first inequality is indicated by the fact that the dual SDP in Eq.~\eqref{Eq:suc_prob_k} is a minimization problem and the last inequality is indicated by the fact that the primal SDP in Eq.~\eqref{Eq:suc_prob} is a maximization problem. Therefore we could conclude $P_{\suc,e}^{\PPT}(\rho,\sigma;2) < P_{\suc}(\rho,\sigma;2)$ for this pair of states. Owing to the inclusion of $\LOCC\subseteq\PPT$, we have $P_{\suc,e}^{\LOCC}(\rho,\sigma;2) < P_{\suc}(\rho,\sigma;2)$, implying that a single ebit does not suffice for optimally discriminating $\rho$ and $\sigma$ with equal prior probability by LOCC.
\end{proof}
\renewcommand{\theproposition}{S\arabic{proposition}}

\end{document}

%% file: pretex.tex
\usepackage{mathtools}
\usepackage[shortlabels]{enumitem}

\usepackage{graphicx,epic,eepic,epsfig,amsmath,latexsym,amssymb,verbatim,color}
 
\usepackage{amsfonts}       
\usepackage{nicefrac}       

\usepackage{amsmath}
\usepackage{bbm}
 
\usepackage[paperwidth=199.8mm,paperheight=297mm,centering,hmargin=20mm,vmargin=2cm]{geometry}
\usepackage[dvipsnames]{xcolor}
\usepackage{framed}
\definecolor{shadecolor}{rgb}{0.9,0.9,0.9}

\usepackage{float}
\usepackage{tikz}
\usetikzlibrary{chains}
\usetikzlibrary{fit}
\usepackage{pgflibraryarrows}		
\usepackage{pgflibrarysnakes}		

\usepackage{epsfig}
\usetikzlibrary{shapes.symbols,patterns} 
\usepackage{pgfplots}

\usepackage[strict]{changepage}

\usepackage[marginal]{footmisc}
\usepackage{url}
\usepackage{theorem}

\newtheorem{definition}{Definition}
\newtheorem{proposition}{Proposition}
\newtheorem{lemma}[proposition]{Lemma}

\newtheorem{theorem}[proposition]{Theorem}

\newtheorem{corollary}[proposition]{Corollary}

\def\squareforqed{\hbox{\rlap{$\sqcap$}$\sqcup$}}
\def\qed{\ifmmode\squareforqed\else{\unskip\nobreak\hfil
\penalty50\hskip1em\null\nobreak\hfil\squareforqed
\parfillskip=0pt\finalhyphendemerits=0\endgraf}\fi}
\def\endenv{\ifmmode\;\else{\unskip\nobreak\hfil
\penalty50\hskip1em\null\nobreak\hfil\;
\parfillskip=0pt\finalhyphendemerits=0\endgraf}\fi}
\newenvironment{proof}{\noindent \textbf{{Proof~} }}{\hfill $\blacksquare$}

\newcounter{remark}
\newenvironment{remark}[1][]{\refstepcounter{remark}\par\medskip\noindent%
\textbf{Remark~\theremark #1} }{\medskip}

\newcounter{example}

\mathchardef\ordinarycolon\mathcode`\:
\mathcode`\:=\string"8000
\def\vcentcolon{\mathrel{\mathop\ordinarycolon}}
\begingroup \catcode`\:=\active
  \lowercase{\endgroup
  \let :\vcentcolon
  }

\RequirePackage[framemethod=default]{mdframed}
\newmdenv[skipabove=7pt,
skipbelow=7pt,
backgroundcolor=darkblue!15,
innerleftmargin=5pt,
innerrightmargin=5pt,
innertopmargin=5pt,
leftmargin=0cm,
rightmargin=0cm,
innerbottommargin=5pt,
linewidth=1pt]{tBox}

\newmdenv[skipabove=7pt,
skipbelow=7pt,
backgroundcolor=red!15,
innerleftmargin=5pt,
innerrightmargin=5pt,
innertopmargin=5pt,
leftmargin=0cm,
rightmargin=0cm,
innerbottommargin=5pt,
linewidth=1pt]{rBox}

\newmdenv[skipabove=7pt,
skipbelow=7pt,
backgroundcolor=blue2!25,
innerleftmargin=5pt,
innerrightmargin=5pt,
innertopmargin=5pt,
leftmargin=0cm,
rightmargin=0cm,
innerbottommargin=5pt,
linewidth=1pt]{dBox}
\newmdenv[skipabove=7pt,
skipbelow=7pt,
backgroundcolor=darkkblue!15,
innerleftmargin=5pt,
innerrightmargin=5pt,
innertopmargin=5pt,
leftmargin=0cm,
rightmargin=0cm,
innerbottommargin=5pt,
linewidth=1pt]{sBox}
\definecolor{darkblue}{RGB}{0,76,156}
\definecolor{darkkblue}{RGB}{0,0,153}
\definecolor{blue2}{RGB}{102,178,255}
\definecolor{darkred}{RGB}{195,0,0}
\definecolor{shadecolor}{gray}{0.9}
\definecolor{honeydew}{rgb}{0.94, 1.0, 0.94}

\newcommand{\nc}{\newcommand}
\nc{\rnc}{\renewcommand}
\nc{\lbar}[1]{\overline{#1}}
\nc{\bra}[1]{\langle#1|}
\nc{\ket}[1]{|#1\rangle}
\nc{\ketbra}[2]{|#1\rangle\!\langle#2|}
\nc{\braket}[2]{\langle#1|#2\rangle}

\nc{\proj}[1]{| #1\rangle\!\langle #1 |}
\nc{\avg}[1]{\langle#1\rangle}
\nc{\rank}{\operatorname{Rank}}
\nc{\smfrac}[2]{\mbox{$\frac{#1}{#2}$}}
\nc{\tr}{\operatorname{Tr}}
\nc{\ox}{\otimes}
\nc{\dg}{\dagger}
\nc{\dn}{\downarrow}
\nc{\cA}{{\cal A}}
\nc{\cB}{{\cal B}}
\nc{\cC}{{\cal C}}
\nc{\cD}{{\cal D}}
\nc{\cE}{{\cal E}}
\nc{\cF}{{\cal F}}
\nc{\cG}{{\cal G}}
\nc{\cH}{{\cal H}}
\nc{\cI}{{\cal I}}
\nc{\cJ}{{\cal J}}
\nc{\cK}{{\cal K}}
\nc{\cL}{{\cal L}}
\nc{\cM}{{\cal M}}
\nc{\cN}{{\cal N}}
\nc{\cO}{{\cal O}}
\nc{\cP}{{\cal P}}
\nc{\cQ}{{\cal Q}}
\nc{\cR}{{\cal R}}
\nc{\cS}{{\cal S}}
\nc{\cT}{{\cal T}}
\nc{\cU}{{\cal U}}
\nc{\cV}{{\cal V}}
\nc{\cX}{{\cal X}}
\nc{\cY}{{\cal Y}}
\nc{\cZ}{{\cal Z}}
\nc{\cW}{{\cal W}}
\nc{\csupp}{{\operatorname{csupp}}}
\nc{\qsupp}{{\operatorname{qsupp}}}
\nc{\var}{{\operatorname{var}}}
\nc{\rar}{\rightarrow}
\nc{\lrar}{\longrightarrow}
\nc{\polylog}{{\operatorname{polylog}}}
\nc{\wt}{{\operatorname{wt}}}
\nc{\av}[1]{{\left\langle {#1} \right\rangle}}
\nc{\supp}{{\operatorname{supp}}}

\nc{\argmin}{{\operatorname{argmin}}}

\def\x{\xi}

\nc{\RR}{{{\mathbb R}}}
\nc{\CC}{{{\mathbb C}}}
\nc{\FF}{{{\mathbb F}}}
\nc{\NN}{{{\mathbb N}}}
\nc{\ZZ}{{{\mathbb Z}}}
\nc{\PP}{{{\mathbb P}}}
\nc{\QQ}{{{\mathbb Q}}}
\nc{\UU}{{{\mathbb U}}}
\nc{\EE}{{{\mathbb E}}}
\nc{\id}{{\operatorname{id}}}

\nc{\CHSH}{{\operatorname{CHSH}}}
\nc{\suc}{{\text{suc}}}

\nc{\be}{\begin{equation}}
\nc{\ee}{{\end{equation}}}
\nc{\bea}{\begin{eqnarray}}
\nc{\eea}{\end{eqnarray}}
\nc{\<}{\langle}
\rnc{\>}{\rangle}
\nc{\rU}{\mbox{U}}

\nc{\ob}[1]{#1}

\nc{\SEP}{{\text{\rm SEP}}}
\nc{\NS}{{\text{\rm NS}}}
\nc{\LOCC}{{\text{\rm LOCC}}}
\nc{\PPT}{{\text{\rm PPT}}}
\nc{\EXT}{{\text{\rm EXT}}}

\nc{\Sym}{{\operatorname{Sym}}}

\nc{\ERLO}{{E_{\text{r,LO}}}}
\nc{\ERLOCC}{{E_{\text{r,LOCC}}}}
\nc{\ERPPT}{{E_{\text{r,PPT}}}}
\nc{\ERLOCCinfty}{{E^{\infty}_{\text{r,LOCC}}}}
\nc{\Aram}{{\operatorname{\sf A}}}

\usepackage{tikz}
\usepackage{tcolorbox}
\usepackage{hyperref}
\hypersetup{colorlinks=true,citecolor=blue,linkcolor=blue,filecolor=blue,urlcolor=blue,breaklinks=true}
\makeatletter
\def\grd@save@target#1{%
  \def\grd@target{#1}}
\def\grd@save@start#1{%
  \def\grd@start{#1}}
\tikzset{
  grid with coordinates/.style={
    to path={%
      \pgfextra{%
        \edef\grd@@target{(\tikztotarget)}%
        \tikz@scan@one@point\grd@save@target\grd@@target\relax
        \edef\grd@@start{(\tikztostart)}%
        \tikz@scan@one@point\grd@save@start\grd@@start\relax
        \draw[minor help lines,magenta] (\tikztostart) grid (\tikztotarget);
        \draw[major help lines] (\tikztostart) grid (\tikztotarget);
        \grd@start
        \pgfmathsetmacro{\grd@xa}{\the\pgf@x/1cm}
        \pgfmathsetmacro{\grd@ya}{\the\pgf@y/1cm}
        \grd@target
        \pgfmathsetmacro{\grd@xb}{\the\pgf@x/1cm}
        \pgfmathsetmacro{\grd@yb}{\the\pgf@y/1cm}
        \pgfmathsetmacro{\grd@xc}{\grd@xa + \pgfkeysvalueof{/tikz/grid with coordinates/major step}}
        \pgfmathsetmacro{\grd@yc}{\grd@ya + \pgfkeysvalueof{/tikz/grid with coordinates/major step}}
        \foreach \x in {\grd@xa,\grd@xc,...,\grd@xb}
        \node[anchor=north] at (\x,\grd@ya) {\pgfmathprintnumber{\x}};
        \foreach \y in {\grd@ya,\grd@yc,...,\grd@yb}
        \node[anchor=east] at (\grd@xa,\y) {\pgfmathprintnumber{\y}};
      }
    }
  },
  minor help lines/.style={
    help lines,
    step=\pgfkeysvalueof{/tikz/grid with coordinates/minor step}
  },
  major help lines/.style={
    help lines,
    line width=\pgfkeysvalueof{/tikz/grid with coordinates/major line width},
    step=\pgfkeysvalueof{/tikz/grid with coordinates/major step}
  },
  grid with coordinates/.cd,
  minor step/.initial=.2,
  major step/.initial=1,
  major line width/.initial=2pt,
}
\makeatother

\usepackage{thmtools}
\usepackage{thm-restate}
\usepackage{etoolbox}
\makeatletter
\def\problem@s{}
\newcounter{problems@cnt}

\newcommand{\allproblems}{\problem@s}
\makeatother